\documentclass{comsoc2014}
\usepackage{verbatim}
\usepackage{hyperref}
\usepackage{amsmath,amsthm,amsfonts,amssymb}
\usepackage{thmtools}
\usepackage{enumitem}
\usepackage{datetime}
\usepackage{multirow}

\usepackage{algorithm}
\usepackage{algpseudocode}

\title{Stable Invitations}
\author{Hooyeon Lee and Yoav Shoham\\$ $ \\
}

\newcommand{\ASIP}{{\normalfont \textsf{ASIP}}} 
\newcommand{\GASP}{{\normalfont \textsf{GASP}}} 
\newcommand{\GSIP}{{\normalfont \textsf{GSIP}}}

\newcommand{\SPK}{{\normalfont \textsf{SPK}}} 
\newcommand{\INC}{{\normalfont \textsf{INC}}} 
\newcommand{\DEC}{{\normalfont \textsf{DEC}}}

\setitemize[0]{itemsep=-1pt}
\theoremstyle{plain}

\newtheorem*{Thm*}{Theorem}
\newtheorem*{Lemma*}{Lemma}

\newtheorem{Thm}{Theorem}
\newtheorem{Lemma}{Lemma}

\theoremstyle{definition}
\declaretheorem[style=definition,name=Example]{Exmp}
\declaretheorem[style=definition,name=Definition]{Def}

\makeatletter
\@addtoreset{proofpart}{Thm}
\makeatother

\begin{document}


\begin{abstract}
We consider the situation in which an organizer is trying to convene
an event, and needs to choose a subset of agents to be invited.
Agents have preferences over how many attendees should be at the event
and possibly also who the attendees should be.
This induces a stability requirement: All invited agents should prefer
attending to not attending, and all the other agents should not regret
being not invited.
The organizer's objective is to find the invitation of maximum size
subject to the stability requirement.
We investigate the computational complexity of finding the maximum
stable invitation when all agents are truthful, as well as the
mechanism design problem when agents may strategically misreport their
preferences.
\end{abstract}


\section{Introduction}  \label{sec:Intro}
Imagine an event organizer trying to convene an event -- for example,
a fundraiser.
Let us assume for now (as we will do in most of the paper) that the
time and venue for the event are fixed,
and that the only remaining decision for the organizer to make is whom
to invite among a set of agents.
An \emph{invitation} is simply defined to be a subset of agents.
The goal of the organizer is to find maximum invitations (for example,
in order to maximize the total donations), but the potential invitees have 
their own preferences over how many attendees there should be at the event 
and possibly also who the potential attendees should be.
For example, a given donor may not want to attend if few attendees 
show up, but she may also want the event not to be overly crowded. 
Another donor may want to attend the event only
if her friends attend and her business competitor does not.

We first consider agents with anonymous preferences over invitations
-- agents only care about how many attendees are at the event (but not
the identities of attendees).
An invitation is \emph{stable} if all invitees prefer
attending to not attending and if no uninvited person wishes she had
been invited. Stability is obviously desirable, but a stable
invitation may not exist in general. This naturally raises the
question of how hard it is to determine whether it exists for a
given setting, and if it does, what the maximum stable invitation is.
These questions take an extra meaning in the
strategic case, in which agents may misreport their preferences. Can
the organizer incentivize the agents to disclose their true preferences? 
We call this problem the \emph{Anonymous Stable Invitation Problem} (\ASIP). 
If we assume truthful agents, we have an algorithm design problem, and
we report positive results in this case. If we assume strategic
agents, we have a mechanism design problem, and
we report an impossibility result in general as well as positive
results for a special case of the problem.

We then relax the assumption of anonymous preferences and allow agents
to specify their constraints over the identities of attendees, 
in addition to their preferences over the number of attendees. 
Formally, we allow each agent to specify her acceptance set of agents
and rejection set of agents such that she is willing to attend the
event only if everyone in her acceptance set also attends, no one
in her rejection set attends, and the number of attendees is acceptable to her.
As before, an invitation is \emph{stable} if all invitees prefer
attending to not attending and if no uninvited person wishes she had
been invited.
We again ask the following natural questions: How hard it is to
determine whether a stable invitation exists for a given instance,
and if it exists whether one can efficiently find the maximum stable invitation. 
We call this problem the \emph{General Stable Invitation Problem} (\GSIP). 
In the non-strategic case we show that one can efficiently find a maximum
stable invitation in some cases, but in general the problem is
NP-hard. In the strategic case, an impossibility for the anonymous
case implies the same impossibility for the general case. But we
also show an impossibility result for the case in which agents only
care about the identities of attendees (and not about the number of
attendees).

The rest of the paper is organized as follows. In Section~\ref{sec:Related} we discuss previous work in the literature. In Section~\ref{sec:ASIP} we formally define the Anonymous Stable Invitation Problem, and investigate both the non-strategic and strategic cases.  In Section~\ref{sec:GSIP} we relax the assumption of anonymous preferences and formally define the General Stable Invitation Problem. We then investigate the non-strategic and strategic cases of the general problem. In Section~\ref{sec:Discussion} we discuss the contribution of this work and directions for future work.

\section{Related Work} \label{sec:Related}


The problems we consider in this work can be viewed as a group scheduling 
problem in which the goal is to find an agreeable outcome for a group of 
agents, subject to certain constraints and objectives. 
Ephrati et al.~\cite{ephrati1994non} tackle incentive issues with strategic 
agents in a group scheduling setting, and propose three monetary-based meeting systems.
Their systems are extensions of the Vickrey-Clarke-Groves mechanism, 
in which invitees can bid their preferences using monetary ``points."
They assume that a central host has an access to the calendars of invitees, 
but in this work we do not make this assumption. In addition we do not consider 
monetary-based mechanisms in this work.

Another closely related work is done by Lu and 
Boutilier~\cite{GroupPurchasing}. Lu and Boutilier consider the problem of 
group purchasing, in which the sellers provide volume discounts for a large 
group of buyers and the objective is to find a (Nash) stable matching between 
the buyers and sellers that maximizes social welfare. Although the settings 
are quite different, it is worth noting that the problems we discuss in this 
work can be generalized to stable matching problems (we discuss such 
generalization in Section~\ref{sec:Discussion}). To differentiate from the work 
of Lu and Boutilier, we emphasize that we consider both truthful and strategic 
agents and we do not consider coalitions in this work, while they assumed 
truthful agents and considered coalitions among agents.

The most closely related work of which we are aware is done by Darmann 
et al.~\cite{GASP2012WINE}. 
There the authors consider the \emph{Group Activity Selection Problem} (\GASP), 
in which the objective is to find an assignment of agents to activities where 
agents are assumed to have anonymous preferences over activities as well as 
number-of-participants.  We remark that the 
\emph{Anonymous Stable Invitation Problem} (\ASIP\ in Section~\ref{sec:ASIP}) 
can be viewed as a sub-class of \GASP\ with a single activity.\footnote{In particular, our easiness result shown in Section~\ref{sec:ASIP-algorithm} can be derived from the work of Darmann et al. 
However we remark that other easiness and hardness results of ours are not 
directly implied by the work of Darmann et al. }
Yet there are several differences between our work and the work by 
Darmann et al.  First, our main results are anchored in the 
\emph{General Stable Invitation Problem} (\GSIP\ in Section~\ref{sec:GSIP}) 
where agents no longer have anonymous preferences.  Second, while Darmann 
et al.\ provides various hardness results for \GASP, they do not directly 
imply similar hardness results for \GSIP. Finally, in both \ASIP\ and \GSIP, 
we consider strategic agents and strategy-proof mechanisms, while Darmann 
et al.\ only considers non-strategic (truthful) agents for \GASP. 

We also remark that the Stable Invitations Problems (both \ASIP\ and \GSIP) 
are closely related to hedonic games.\footnote{Much work has been devoted to analyzing optimality and stability of various hedonic coalition structures (see \cite{BogomolnaiaJackson}, \cite{DrezeGreenberg}, and \cite{AzizBrandl}). In this work we take a solution concept of stability for granted and focus on computational and game-theoretic complexity of finding a solution. 
}
Ballester~\cite{NP_hedonic} provides a number of computational complexity results (in fact, hardness results) 
for finding a core-stable, Nash-stable, or individually rational outcome 
in hedonic games and anonymous hedonic games.  These hardness results do not 
imply similar hardness results for \GSIP\ for two reasons. First, the hardness 
results on anonymous hedonic games do not hold for \GSIP\ because \GSIP\ is not 
equivalent to anonymous hedonic games. Second, while an instance of \GSIP\ can 
be transformed into a (non-anonymous) hedonic game in a na\"{i}ve manner by 
listing all possible coalitions, this transformation blows up the size of an 
input instance representation exponentially.\footnote{In contrast, any 
instance of \ASIP\ can be transformed into a (non-anonymous) hedonic game, 
and it is described concisely in the work of Darmann et al. However, since we 
provide an easiness result for \ASIP\ in this work, Bellester's hardness 
results are not applicable in \ASIP. }  
Therefore we remark that our easiness and hardness results discussed 
in this work are not a direct derivative of the work by Darmann et al. or Ballester.

\section{Anonymous Stable Invitation Problem} \label{sec:ASIP}

In this section we formally define the \emph{Anonymous Stable Invitation Problem} (\ASIP), and investigate the non-strategic case and the strategic case of \ASIP. 

\subsection{Definitions and Notation}  \label{sec:asip_def}

\begin{Def} \label{def:ASIP}
An instance of the \emph{Anonymous Stable Invitation Problem} (\ASIP) is a pair $(N, P)$ where $N = \{a_1, a_2, \dots, a_n\}$ is a set of $n$ agents and $P$ is an $n$-tuple of preferences of agents where $P = (P_1, P_2, \dots, P_n)$.  For each agent $a_i$, we define $P_i$ to be a total preorder ($\succeq_i$) on the set of outcomes, $X = \{0, 1, 2, \dots, n\}$ where an outcome $x\in (X\setminus\{0\})$ denotes the number of attendees and $x = 0$ denotes the outside option of not attending. 
For any $x_1, x_2\in (X \setminus\{0\})$, $x_1 \succeq_i x_2$ is interpreted as agent $a_i$ weakly preferring attending the event if $x_1$ attendees are present (including herself) to attending if $x_2$ attendees are present (including herself).
We use $\succ_i$ and $\sim_i$ to denote the induced strict preferences and indifference relations, respectively. 
We drop the subscript ($i$) if it is clear from the context.
\end{Def}
We assume that for each agent $a_i$ and each $x\in (X\setminus\{0\})$, either $x \succ_i 0$ or $0 \succ_i x$. In words, no agent is indifferent between her outside option of not attending and any other outcome. This assumption is made for convenience and does not change our technical results. 

If the organizer na\"{\i}vely invites everyone with the goal of maximizing attendance, some invitees might prefer not to attend -- this leads to individual rationality. On the other hand, if the organizer leaves out some agents with the hope of pleasing the rest, this might upset some of those not invited -- this leads to envy-freeness.  We now formally define stable invitations. 
\begin{Def} \label{def:invitation}
An \emph{invitation} $S$ for an instance $(N, P)$ is a subset of $N$, and is interpreted as the organizer inviting the agents in $S$. 
An invitation $S$ is said to be \emph{individually rational} if for every agent $a_i\in S$ it holds that $|S| \succ_i 0$. 
An invitation $S$ is said to be \emph{envy-free}\footnote{Another notion of envy can be defined as an agent who is not invited but desires the spot of another agent who is invited. We do not consider this notion in this work, but it is an interesting direction for future work to see how our technical results would change. } if for every agent $a_i\not\in S$ it holds that $|S \cup \{a_i\}| = (|S| + 1) \prec_i 0$. 
An invitation is \emph{stable} if it is both individually rational and envy-free.
\end{Def}

For each agent $a_i$, we can naturally induce from $P_i$ her preference over the set of all invitations because the preference between any two invitations is induced by the cardinality of each invitation.  We overload our notation ($\sim_i, \succeq_i, \succ_i$) for the induced preferences over invitations. Formally, for any invitation $S$ with $a_i\not\in S$, we define that $S \sim_i 0$ because $a_i$ is not invited.  Between any invitation $S$ with $a_i\in S$ and another invitation $S'$ with $a_i \in S'$, the preference between $S$ and $S'$ is induced by the preference between $|S|$ and $|S'|$ (be it $\sim_i, \succeq_i, \succ_i$). The preference between $S$ with $a_i\in S$ and $S'$ with $a_i\not\in S'$ is induced by the preference between $|S|$ and $0$ (because $S' \sim_i 0$).
The preferences of agents over invitations are now well-defined. 

Let us define a special class of preferences of agents, called single-peaked preferences. 
An agent is said to have a single-peaked preference if the set of outcomes that she strictly prefers to not attending is single-peaked.
\begin{Def}\label{def:preferenceTypes}
Given an instance $(N, P)$, let $Y_i = \{x\in X : x \succ_i 0 \}$ be a subset of outcomes that $a_i$ prefers to $0$. 
$a_i$ is said to have a single-peaked preferences (\SPK) 
if there exist outcomes $l_i \in Y_i$ and $h_i \in Y_i$ such that $x \in Y_i$ if and only if $l_i \leq x \leq h_i$ for any $x\in X$ and there exists an ideal outcome $o_i \in Y_i$ such that 
if $k_1\leq k_2 \leq o_i$ or $k_1 \geq k_2 \geq o_i$, then $k_2 \succeq_i k_1$ for all $k_1, k_2 \in Y$.\footnote{Note that $l_i, h_i, o_i$ may or may not be different.}
If $Y_i = \emptyset$, then we set $o_i = 0$ for convenience.
\end{Def} 

There are two important special cases of SPK-preferences: Increasing preferences (\INC-
preferences) and decreasing preferences (\DEC-preferences). We say that agent $a_i$ has an \INC-preference (\DEC-preference, respectively) if $h_i = o_i = n$ (if $l_i = o_i = 1$, respectively). 
Although we assume that agents can have arbitrary preferences, we will see that \INC-preferences and \DEC-preferences play an important role in the strategic case in Section~\ref{sec:Mechanism}. 


\subsection{Examples}

\begin{Exmp}[Stable invitations are not unique] \label{eg:motivatingExample}
Let us consider three agents $N = \{a_1,a_2, a_3\}$ and the set of outcomes $X = \{0, 1, 2, 3\}$. 
Each agent's preference ordering over $X$ is given by:
\begin{table*}[!h]
	\centering
\begin{tabular}{ccc}
  $P_1: 1 \succ 0 \succ 2 \sim 3$ ,
& $P_2: 1 \succ 0 \succ 2 \sim 3$ ,
&	$P_3: 0 \succ 1 \sim 2 \sim 3$. \\
\end{tabular}
\end{table*}

Note that $a_1$ and $a_2$ have \DEC-preferences with $o_1=o_2 = 1$ and that $a_3$ is unwilling to attend with $o_3 = 0$. Since $a_3$ is not willing to attend, any invitation that contains $a_3$ is not individually rational. Among the remaining four possible invitations with $a_1$ and $a_2$, the empty invitation (i.e., $S = \emptyset$) is not envy-free and the invitation $S = \{a_1, a_2\}$ is not individually rational. 
The two stable invitations are $\{a_1\}$ and $\{a_2\}$. 
\end{Exmp}
\begin{Exmp}[A stable invitation may not exist] \label{eg:noStableSet}
Let us consider two agents $N = \{a_1, a_2\}$ whose preference orderings are: 

\begin{table*}[!h]
	\centering
\begin{tabular}{cc}
  $P_1: 1 \succ 0 \succ 2$ ,
& $P_2: 2 \succ 0 \succ 1$. \\
\end{tabular}
\end{table*}

Note that $a_1$ has a \DEC-preference ($o_1 = l_1 = h_1 = 1$) and $a_2$ has an \INC-preference ($o_2 = l_2 = h_2 = 2$). 
The empty invitation (i.e., $S = \emptyset$) is not envy-free because $a_1$ prefers attending the event. The invitation $S = \{a_1\}$ not envy-free, $S = \{a_2\}$ is not individually rational, and the full invitation (i.e., $S = N$) is not individually rational. Therefore there is no stable invitation.
\end{Exmp}


\subsection{The Non-strategic Case} \label{sec:Algorithm}
We first investigate the non-strategic case of \ASIP\ with truthful agents. We assume that agents can have arbitrary preferences over the size of invitations. 

\subsubsection{Easiness Result}  \label{sec:ASIP-algorithm}
We present an efficient algorithm that determines whether a stable invitation exists and finds a maximum stable invitation (if it exists), when given any instance of \ASIP. Note that a maximum stable invitation needs not be unique. 

\begin{algorithm} 
	\caption{Algorithm for Finding Maximum Stable Invitation}
	\label{alg:stableScheduling}
\begin{algorithmic}[1]
	\State \textbf{Input:} $(N, P = \{P_1, P_2, \dots, P_n\})$
	\For{$k := n, n-1, n-2, \dots, 2, 1$}
		\State $Z \gets \{a_i \in N : k \succ_i 0 \}$
		\If{$|Z| \geq k$}
			\State $S \gets Z$
			\While{$|S| > k$ and $\exists a_i \in S$ such that $0 \succ_i (k+1)$}
				\State Pick any $a_i \in S$ such that $0 \succ_i (k+1)$
				\State $S \gets S \setminus \{a_i\}$
			\EndWhile
			\If{$|S| = k$ and $\{a_j \not\in S : (k+1) \succ_j 0 \} = \emptyset$}
				\State \textbf{return} $S$
			\EndIf
		\EndIf
	\EndFor
	\If{$\{a_i \in N : 1 \succ_i 0 \} = \emptyset$}
		\State \textbf{return} $S = \emptyset$
	\EndIf
	\State \textbf{return} \textsf{NONE} 
\end{algorithmic}
\end{algorithm}

Our algorithm is described in Algorithm~\ref{alg:stableScheduling}, and it works as follows: The algorithm checks whether a stable invitation of size $k$ exists by iterating $k$ in decreasing order (line 2). For some fixed $k$, the algorithm finds a subset $Z$ of $N$ such that $Z$ contains all the agents to whom $k \succ_i 0$ (line 3). The algorithm then checks if $Z$ contains at least $k$ agents (line 4). If it does, then it obtains a subset $S$ of $Z$ by removing the agents to whom $(k+1)$ is not an acceptable outcome (lines 5-9). The algorithm returns $S$ if $|S| = k$ and it is envy-free (lines 10-12).  If the algorithm does not return in line 11 for any $k$, then it checks weather the empty invitation is stable; if it is, then it returns the empty invitation, otherwise it determines that no stable invitations exist (lines 15-18). Note that our algorithm assumes arbitrary preferences of agents. 

Theorem~\ref{thm:asip_algo} states our easiness result for the non-strategic case of \ASIP. (A formal proof can be found in a long version of this paper.)
\begin{Thm} \label{thm:asip_algo}
Given an instance of \ASIP, Algorithm~\ref{alg:stableScheduling} terminates in polynomial time. If a stable invitation exists the algorithm produces a maximum one, otherwise it determines that no stable invitation exists.
\end{Thm}

\subsubsection{An Extension to Multiple Alternatives for Time} \label{sec:algo_extension_multiple}
Although we assumed that the time for the event is set in advance, there may exist many alternatives for the time of the event. 
We can extend \ASIP\ such that each agent can specify her preferences over the pairs of (time, number-of-attendees). The organizer then chooses both the time and the invitation, and the goal is to maximize the size of the invitation subject to the same stability constraints. 
In this setting, the organizer can simply execute Algorithm~\ref{alg:stableScheduling} iteratively for each time-alternative, and find a maximum stable invitation with respect to the fixed time alternative.
Therefore our easiness results can be naturally extended to the case with multiple time-alternatives. 
Yet we remark that this extension complicates the problem more severely if we assume strategic agents, as we discuss in Section~\ref{sec:mech_extension_multiple}.

\subsection{The Strategic Case}\label{sec:Mechanism}
The following example shows how an agent is incentivized to misreport her preferences.

\begin{Exmp}[An agent may be incentivized to misreport.] \label{eg:strategicAgents}
Suppose there are $3$ agents whose true preferences are:

\begin{table*}[!h]
	\centering
\begin{tabular}{ccc}
  $P_1: 1 \succ 2 \succ 3 \succ 0$ ,
& $P_2: 3 \succ 0 \succ 1 \sim 2$ ,
&	$P_3: 3 \succ 0 \succ 1 \sim 2$. \\
\end{tabular}
\end{table*}

Notice that agent $a_1$ has a \DEC-preference (with $l_1 = o_1 = 1$ and $h_1 = 3$) while the other two agents have \INC-preferences (with $o_i = l_i = h_i = 3$ for $i\in \{2,3\}$).

If all agents are truthful, the full invitation ($S = \{a_1, a_2, a_3\}$) is maximum and stable. 
Suppose agent $a_1$ misreports her preference ordering as $\hat{P_1}$:
\begin{table*}[!h]
	\centering
\begin{tabular}{c}
  $\hat{P_1}: 1 \succ 0 \succ 2 \sim 3$.
\end{tabular}
\end{table*}

Given the preferences $\{\hat{P_1}, P_2, P_3\}$, the maximum stable invitation is $\hat{S} = \{a_1\}$; the full invitation $S$ is seemingly no longer stable.
Since $a_1$ prefers $\hat{S}$ to ${S}$ (as $1 \succ_1 3$), she has an incentive to misreport her preference ordering. 
\end{Exmp}

\subsubsection{Definitions and Notation} \label{sec:def_notation_asip_mechanism}
Let us formally define a mechanism in the context of \ASIP\ with strategic agents. 
Although we only provide a definition of a deterministic mechanism here, we discuss how one can generalize the definition to a randomized mechanism at the end of this section. 

\begin{Def} \label{def:mechanism}
Given an instance $(N, P)$ of \ASIP, we define $V_i$ (the set of available actions to $a_i$) to be the set of all preferences over $X$ where $X = \{0, 1, 2, \dots, n\}$ is the set of outcomes. 
A (deterministic) \emph{mechanism} is a pair $(V, Z)$ where $V = (V_1 \times \cdots \times V_n)$ is the set of action profiles of all agents and $Z: V \mapsto U$ is a mapping from each action profile to an invitation in $U$ where $U = 2^{N}$.
Let $V_{-i} = (V_1 \times \cdots \times V_{i-1} \times V_{i+1} \times \cdots \times V_{n})$ be the set of action profiles available to all agents but agent $a_i$. A mechanism $(V, Z)$ is said to be \emph{strategy-proof} if for all $a_i\in N$ it holds that $Z(P_i, v_{-i}) \succeq_i Z(v_i, v_{-i})$ for all $v_i \in V_i$ and $v_{-i} \in V_{-i}$.
\end{Def}


\subsubsection{Impossibility Result}
We start with an impossibility result. Theorem~\ref{thm:impossibility} states that strategy-proofness of a mechanism and capability of finding a stable invitation are incompatible (let alone finding a maximum one). 

\begin{Thm} \label{thm:impossibility}
When given an instance of \ASIP, there is no strategy-proof mechanism that finds a stable invitation (with respect to true preferences of agents), even if it exists. 
\end{Thm}
\begin{proof} 
Consider two agents with preferences:
\begin{table*}[!h]
	\centering
\begin{tabular}{cc}
  $P_1: 1 \succ 0 \succ 2$ ,
& $P_2: 1 \succ 0 \succ 2$. \\
\end{tabular}
\end{table*}

Both agents prefer to attend the event alone, and there are only two stable invitations: $S_1 = \{a_1\}$ and $S_2 = \{a_2\}$.
Let $V_i$ be the set of all action profiles of agent $a_i$, and let $V = V_1 \times V_2$. 
Consider any mapping $Z$ from $V$ to $U = 2^{N}$.
We will show that if $(V, Z)$ is a strategy-proof mechanism, $Z$ must map $(P_1, P_2)$ to neither $S_1$ nor $S_2$, which means that the mechanism does not find a stable invitation.

If $Z$ maps $(P_1, P_2)$ to $S_1$, then $a_2$ has an incentive to report her preference untruthfully as $v_2 \in V_2$ instead:
\begin{table*}[!h]
	\centering
\begin{tabular}{c}
  $v_2: 1 \sim 2 \succ 0$.
\end{tabular}
\end{table*}

Given $(P_1, v_2)$, the only stable invitation is now $S_2=\{a_2\}$ because $\{a_1, a_2\}$ is not individually rational, while neither $\emptyset$ nor $S_1=\{a_1\}$ is envy-free. 
Notice that $a_2$ strictly prefers $S_2$ over $S_1$ (because $S_1 \sim_2 0$ as $a_2$ is not invited) and therefore $a_2$ has an incentive to deviate from $P_2$ to $v_2$.  Similarly, if $Z$ maps $(P_1, P_2)$ to $S_2$, then $a_1$ has an incentive to deviate from $P_1$ to an untruthful preference ordering (namely, $v_1 : 1 \sim 2 \succ 0$). Therefore if a mechanism $(V, Z)$ is strategy-proof, then $Z$ must map $(P_1, P_2)$ to neither $S_1$ nor $S_2$; yet these two invitations are the only stable invitations given $(P_1, P_2)$.
Thus there is no strategy-proof mechanism that can find a stable invitation for this instance.
\end{proof}
The intuition behind the example used in our proof is simple.  While the event organizer is trying to maximize attendance, agents prefer minimizing attendance. This conflict of interests can make it impossible for the organizer to find a stable invitation while ensuring strategy-proofness. 

Note that all of $\{P_1, P_2, v_1, v_2\}$ used in our proof are \DEC-preferences. Even if we limit both $V_1$ and $V_2$ to be the set of \DEC-preferences over $X$ (instead of arbitrary preferences over $X$), which limits the ability of agents to manipulate the mechanism, this impossibility result still holds. 
This strengthens our impossibility result.

Given this impossibility result, we ask the following natural question: Is it possible for a manipulator to determine efficiently, which action of his would lead to a more favorable outcome (than his truthful action), provided that he knows preferences of all other agents? The following lemma informally answers this question while we omit details due to space.
We also note that this lemma is still applicable when we extend \ASIP\ to \GSIP\ in next section.
\begin{Lemma}[Informal] \label{lemma:gsip_manipulation}
	There exists a polynomial time algorithm, given an instance $(N, P)$ of \GSIP\ and a mechanism $(V, Z)$, that decides in polynomial time whether there exists a certain preference ordering $v_i \in V_i$ of $a_i$ such that $a_i$ (strictly) prefers $Z(v_i, P_{-i})$ to $Z(P_i, P_{-i})$.
\end{Lemma}

\subsubsection{Easiness Result (Strategy-proof Mechanism)}
Let us now consider a special instance of \ASIP\ in which all agents have \INC-preferences; we call such instances \INC-instances of \ASIP. Since each agent has an \INC-preference, we can induce the \emph{minimum threshold}, $l_i$, of agent $a_i$ such that $k \succ_i 0$ if and only if $k \geq l_i$.
 In the case where $a_i$ has no acceptable outcome (i.e., $0 \succ_i x$ for all $x\in (X \setminus\{0\} )$), we simply define $l_i = n+1$. 
Let $l_i$ denote the \emph{true induced minimum threshold} of $a_i$ given $P_i$ and $L_i$ the \emph{reported induced minimum threshold} of $a_i$ given $v_i$ (which may differ from $P_i$). 
Without loss of generality, we can assume that agents report $\{L_i\}$ to the organizer (instead of reporting preference orderings $\{v_i\}$) and that $L_i$ values are sorted in non-decreasing order.

\begin{algorithm}
	\caption{Strategy-proof Mechanism for Finding Maximum Stable Invitation}
	\label{alg:mechanism}
\begin{algorithmic}[1]
	\State \textbf{Input:} $(N, \{L_i : a_i \in N\})$
	\For{$k := n, n-1, n-2, \dots, 2, 1$}
		\If{$L_k \leq k$}
			\State \textbf{return} $S = \{a_1, a_2, \dots, a_{k-1}, a_{k}\}$
		\EndIf
	\EndFor
	\State \textbf{return} $S = \emptyset$
\end{algorithmic}
\end{algorithm}
Our (deterministic) mechanism is presented in Algorithm~\ref{alg:mechanism}, and it works as follows: It is given a sorted list of threshold values in non-decreasing order (line 1). The mechanism determines whether a stable invitation of size $k$ exists by iterating in decreasing order (line 2). If a stable invitation of size $k$ exists (line 3), then it returns the invitation $S$ that contains precisely $k$ agents (line 4). If no stable invitation has been found, then mechanism determines that the empty invitation is the only stable invitation, and returns it in line 7.  

\begin{Thm} \label{thm:mechanism} 
The mechanism described in Algorithm~\ref{alg:mechanism} runs in polynomial time, and finds a maximum stable invitation, provided that all agents report truthfully. The mechanism is also strategy-proof. 
\end{Thm}
\begin{proof}[Proof sketch]
We provide our intuition behind the proof.  If agent $a_i$ over-reports with $L_i > l_i$, it can only result in reducing the size of an invitation that the mechanism chooses, but $a_i$ has an \INC-preference and therefore $a_i$ has no incentive to over-report. If $a_i$ under-reports with $L_i < l_i$, either it does not change the invitation that the mechanism chooses or it leads to an invitation whose size is not acceptable to $a_i$ (i.e., the size is still below her threshold, $l_i$). 
A formal proof can be found in a long version of this paper.
\end{proof}

Note that Theorem~\ref{thm:mechanism} implies that an \INC-instance of \ASIP\ always admit a stable invitation because the mechanism always returns an invitation (possibly the empty invitation). 
Contrary to the previous case where the conflict of interests between agents and the organizer lead to an impossibility result, we have obtained a strategy-proof mechanism that also finds a maximum stable invitation when the interests of both parties align in \INC-instances of \ASIP, which is to maximize attendance.

\subsubsection{An Extension to Multiple Alternatives for Time} \label{sec:mech_extension_multiple}
In Section~\ref{sec:algo_extension_multiple}, we discussed the setting with multiple alternatives for the time of the event in the non-strategic case.  In the strategic case, our impossibility result for the general case is immediately implied.
For \INC-instances of \ASIP, we report that even if there are only two alternatives for the time, it is impossible to design a strategy-proof mechanism that finds a stable invitation (we can construct an example that is similar to the one used in the proof of Theorem~\ref{thm:impossibility}).
Our intuition is that, even if the interests of the agents and organizer align to maximize attendance, preferences over alternatives for the time can incentivize agents to misreport in order to lead to a more favorable outcome. In particular, over-reporting ($L_i > l_i$) can give an agent the veto power on certain alternatives for the time, which in turn can lead to a more favorable outcome for the concerned agent.

\subsubsection{Note on Randomized Mechanisms} 
Although we only discussed deterministic mechanisms, our results can be extended to randomized mechanisms. Here is the general outline. First we define $Z$ to be a mapping from $V$ to $\Pi(U)$ where $\Pi(U)$ denotes the set of all probability distributions over $U$. The definition of a strategy-proof mechanism must change accordingly -- we do this by adopting the axioms in the von Neumann-Morgenstern utility theorem~\cite{von1947theory}. We introduce lotteries over invitations and define preferences of agents over lotteries. Given a probability distribution over invitations, one can compute the expected cardinal utility of lotteries. We then define a strategy-proof mechanism analogously to Definition~\ref{def:mechanism}: for each $a_i$, it must hold that the expected utility of $Z(P_i, v_{-i})$ is no less than the expected utility of $Z(v_i, v_{-i})$ for all $v_i\in V_i$ and for all $v_{-i} \in V_{-i}$. We note that our impossibility result given by Theorem~\ref{thm:impossibility} still holds: If $(V, Z)$ is a strategy-proof mechanism, then $Z(P_1, P_2)$ must assign zero probability to both $\{a_1\}$ and $\{a_2\}$, yet these are the only two stable invitations.

\section{General Stable Invitation Problem} \label{sec:GSIP}

We have so far explored the problem of finding a maximum stable invitation in a setting 
where agents are indifferent among the invitations of the same size.
We now remove this assumption and allow agents to specify their preferences over exactly which agents they like or do not like to attend the event, in addition to their preferences over the size of invitations.

\subsection{Definitions and Notation}
Let us formally define the General Stable Invitation Problem (\GSIP) and solution concepts. 

\begin{Def}\label{def:GSIP}
An instance of the \emph{General Stable Invitation Problem} (\GSIP) is a tuple $(N, P, F, R)$ where $N$ and $P$ are defined the same as before (see Definition~\ref{def:ASIP}), $F = (F_1, F_2, \dots, F_n)$ is a collection of acceptance sets where $F_i \subseteq (N\setminus\{a_i\})$ for each $i$, and $R = (R_1, R_2, \dots, R_n)$ is a collection of rejection sets where $R_i \subseteq (N\setminus\{a_i\})$ for each $i$. We interpret $F_i$ and $R_i$ as a constraint such that $a_i$ is willing to attend the event only if all agents in $F_i$ attend and no agent in $R_i$ attends. 
\end{Def}
Given an instance $(N, P, F, R)$ of \GSIP, we say that it is an $(\alpha,\beta)$-instance, where $\alpha = \max_{a_i\in N} |F_i|$ and $\beta = \max_{a_i\in N} |R_i|$. It holds by definition that $0 \leq \alpha, \beta \leq n-1$; in particular empty acceptance sets and rejection sets are allowed in our definition. We will later see that our easiness and hardness results rely on $(\alpha, \beta)$ values. Notice that any \ASIP\ instance is a $(0,0)$-instance of \GSIP, and therefore \ASIP\ is a special case of \GSIP.\footnote{We emphasize that agents still have preferences over outcomes in \GSIP, although our results in this section focus on how acceptance and rejection sets affect easiness and hardness of \GSIP.  }

We define an \emph{invitation} and its stability constraints the same way we did in Section~\ref{sec:ASIP}.
\begin{Def} \label{def:invitation_gsip}
An \emph{invitation} for an instance $(N, P, F, R)$ is a subset $S$ of $N$, and is interpreted as the organizer invites the agents in $S$.
An invitation $S$ is said to be \emph{individually rational} if for every agent $a_i\in S$ it holds that
$F_i \subseteq S$, $R_i\cap S = \emptyset$, and $|S| \succ_i 0$.
An invitation $S$ is said to be \emph{envy-free} if for every agent $a_i\not\in S$, $S' = S \cup \{a_i\}$ is not an individually rational invitation with respect to $a_i$; that is, if at least one of the following holds: $F_i \not\subseteq S'$, $R_i \cap S' \neq \emptyset$, and $|S'| = (|S| + 1) \prec_i 0$.
An invitation is \emph{stable} if it is both individually rational and envy-free.
\end{Def}

For each agent $a_i$, we can naturally induce from $P_i$ the preference of $a_i$ over the set of all invitations, $2^{N}$, in the same manner as we did in Section~\ref{sec:asip_def}. Therefore, the preferences of agents over invitations are well-defined. Note that individual rationality and envy-freeness are properties of an invitation, not preferences of agents, and therefore we define the induced preferences over invitations to be independent of the properties of a solution concept. 

Next we define a special class of preferences, called simple preferences. 
\begin{Def} \label{def:simple_pref_gsip}
Agent $a_i$ is said to have a \emph{simple} preference, if agent $a_i$ strictly prefers any outcome $x \in X$ with $x \neq 0$ to her outside option, $0$ (i.e.\ for all $x\in X$ with $x \neq 0$, $x \succ_i 0$).
\end{Def}
Note that when an agent has a simple preference, she still has a preference ordering over outcomes, but she strictly prefers attending to not attending.

We emphasize that in \GSIP\ agents still have preferences over sizes of invitations.  
However, our hardness and impossibility results for \GSIP\ are provided while assuming that all agents have simple preferences. 
Note that this assumption strengthens our negative result because it directly implies the same negative result for \GSIP\ with arbitrary preferences.
On the other hand, we provide our easiness result for \GSIP\ with arbitrary preferences, which of course implies the same positive result for \GSIP\ with simple preferences. Henceforth we describe an instance of \GSIP\ simply as $(N, F, R)$ by omitting $P$ when we assume that all agents have simple preferences.

\subsection{Examples}

\begin{Exmp}[Stable invitations are not unique] \label{eg:gsip_not_unique}
Suppose there are $4$ agents denoted by $N = \{a_1, a_2, a_3, a_4\}$.
We assume that these agents have simple preferences, and their acceptance sets and rejection sets are given by:

\begin{table*}[!h]
	\centering
\begin{tabular}{cccc}
  $F_1 = \{a_2\}$ ,  & $F_2 = \{a_1\}$ , & $F_3 = \{a_4\}$ , & $F_4 = \{a_3\}$ , \\  
  $R_1 = \{a_3\}$ ,  & $R_2 = \{a_4\}$ , & $R_3 = \{a_1\}$ , & ~$R_4 = \{a_2\}$.~
\end{tabular}
\end{table*}

Agents $a_1$ and $a_2$ have each other in their acceptance sets, while they reject agents $a_3$ and $a_4$, respectively. Similarly, agents $a_3$ and $a_4$ have each other in their acceptance sets, while they reject agents $a_1$ and $a_2$, respectively. 
There are three stable invitations in this example: $\emptyset$, $\{a_1, a_2\}$, and $\{a_3, a_4\}$ (among $2^{|N|} = 16$ possible invitations). The latter two are maximum stable invitations.  One can easily verify that all other invitations are not stable; for example, $S = \{a_1, a_3, a_4\}$ is not individually rational due to agents $a_1$ and $a_3$.
\end{Exmp}

\begin{Exmp}[Stable invitations may not exist] \label{eg:gsip_not_exist}
Suppose there are $3$ agents denoted by $N = \{a_1, a_2, a_3\}$.
We assume that these agents have simple preferences, and their acceptance sets and rejection sets are given by:
\begin{table*}[!h]
	\centering
\begin{tabular}{cccc}
  $F_1=F_2=F_3=\emptyset$ ,  &  $R_1 = \{a_2\}$ ,  & $R_2 = \{a_3\}$ , & $R_3 = \{a_1\}$.
\end{tabular}
\end{table*}

All three agents accept no other agents, while they form a cyclic rejection-relationship.  
It is easy to verify that there is no stable invitation in this example. 
The empty invitation (i.e.\ $S = \emptyset$) is not envy-free because each agent then would rather attend the event.
An invitation with any single agent is not envy-free; for example if $S = \{a_1\}$, then $S$ is not envy-free with respect to agent $a_2$. 
An invitation with any pair of agents is not individually rational; for example if $S = \{a_2, a_3\}$, then $S$ is not individually rational with respect to $a_2$. 
The full invitation (i.e.\ $S = N$) is not individually rational due to each agent's rejection set. 
Therefore this example admits no stable invitation.
\end{Exmp}


\subsection{The Non-strategic Case} \label{sec:gsip_complexity}
The decision problem of \GSIP\ is whether a stable invitation of size at least $k$ exists given an instance of \GSIP. We show that \GSIP\ is NP-hard even if we assume truthful agents with simple preferences. Note that NP-hardness for \GSIP\ implies NP-completeness because \GSIP\ is clearly in NP (i.e.\ one can efficiently check whether a given invitation is stable).

Theorem~\ref{thm:hardness_gsip_1_1} states that the decision problem of \GSIP\ with simple preferences is NP-hard even if the size of all acceptance sets and rejection sets are at most one. Moreover the theorem states that the problem remains NP-hard even if we relax the stability requirement and only seek a maximum individually rational invitation that may not be envy-free. Theorem~\ref{thm:hardness_gsip_0_2} delivers a similar negative result even if we restrict to $(0,2)$-instances of \GSIP\ with simple preferences. (All omitted proofs of this section can be found in a long version of this paper.)

\begin{Thm} \label{thm:hardness_gsip_1_1}
	It is NP-hard to decide whether a $(1, 1)$-instance of \GSIP\ admits a stable invitation of size at least $k$, even if all agents have simple preferences. 
	It is NP-hard to decide whether a $(1, 1)$-instance of \GSIP\ admits an individually rational invitation of size at least $k$, even if all agents have simple preferences.
\end{Thm} 
\begin{Thm} \label{thm:hardness_gsip_0_2}
	It is NP-hard to decide whether a $(0, 2)$-instance of \GSIP\ admits a stable invitation of size at least $k$, even if all agents have simple preferences. 
	It is NP-hard to decide whether a $(0, 2)$-instance of \GSIP\ admits an individually rational invitation of size at least $k$, even if all agents have simple preferences. 
\end{Thm}
Note that NP-hardness of $(a, b)$-instances implies NP-hardness of $(\alpha,\beta)$-instances where $\alpha \geq a$ and $\beta \geq b$ and that NP-hardness for \GSIP\ with simple preferences immediately implies NP-hardness for \GSIP\ with arbitrary preferences.

We now consider the remaining cases of \GSIP\ whose computational complexity is not implied by Theorem~\ref{thm:hardness_gsip_1_1} and \ref{thm:hardness_gsip_0_2}. We know that $(0,0)$-instances of \GSIP\ with arbitrary preferences are solvable in polynomial time because those are instances of \ASIP. In addition if we are given $(1,0)$-instances or $(0,1)$-instances of \GSIP\ with arbitrary preferences, we can find a maximum stable invitation in polynomial time. 

\begin{Thm} \label{thm:gsip_1_0_polytime}
There exists a polynomial time algorithm that finds a maximum stable invitation when given a $(1, 0)$-instance $(N, P, F, R)$ of \GSIP\ (with arbitrary preferences, $P$). 
\end{Thm}
\begin{proof}[Proof sketch]
Given a $(1,0)$-instance $(N, P, F, R)$ of \GSIP, we create a directed graph $G = (V, E)$ such that $V = N$ (each node corresponds to an agent) and $E = \{(a_i, a_j) : a_j \in F_i, \forall a_i \in N \}$. An edge $e = (a_i, a_j)$ has the direction from $a_i$ to $a_j$. Since each node contributes at most one edge, we know that each component of $G$ either is a tree or contains a cycle such that each node on the cycle is the root of a tree. Utilizing this structure of the graph, we can determine for fixed $k$ whether a stable invitation of size $k$ exists in polynomial time by using a dynamic programming algorithm. We omit details of the proof.
\end{proof}
\begin{Thm} \label{thm:gsip_0_1_polytime}
	There exists a polynomial time algorithm that finds a maximum stable invitation when given a $(0, 1)$-instance $(N, P, F, R)$ of \GSIP\ (with arbitrary preferences, $P$). 
\end{Thm}
\begin{proof}[Proof sketch]
Given a $(0,1)$-instance $(N, P, F, R)$ of \GSIP, we create a directed graph $G = (V, E)$ analogous to the graph we constructed in the proof sketch of Theorem~\ref{thm:gsip_1_0_polytime}. We again use a dynamic programming algorithm to determine whether a stable invitation of size $k$ exists in polynomial time. We omit details of the proof.
\end{proof}
We emphasize that our algorithms for Theorems~\ref{thm:gsip_1_0_polytime} and \ref{thm:gsip_0_1_polytime} rely on the restriction that each agent's acceptance set (rejection set, respectively) is limited to singleton sets or empty sets.

Finally, we consider $(\alpha, 0)$-instances of \GSIP\ when $\alpha \geq 2$. 
In this sub-class of \GSIP, the decision problem is in P if we assume simple preferences while it is NP-hard if we assume arbitrary preferences. Recall that all of our results so far did not differ whether we assume simple preferences or arbitrary preferences.
\begin{Thm} \label{thm:hardness_gsip_2_0}
	It is NP-hard to decide whether a $(\alpha, 0)$-instance of \GSIP\ admits a stable invitation of size at least $k$ when agents can have arbitrary preferences and $\alpha \geq 2$.
\end{Thm}
\begin{Lemma} \label{lemma:trivial_poly_gsip}
Given any $(\alpha,0)$-instance $(N, F, R)$ of \GSIP\ with simple preferences, the full invitation is the unique maximum stable invitation.
\end{Lemma}
\begin{proof}
Let $S = N$ be the full invitation given any $(\alpha,0)$-instance $(N, F, R)$ of \GSIP\ with simple preferences.  Because $S$ is the full invitation, we know that $S$ is envy-free by definition. 
For all $a_i \in S$, we know that $F_i\subseteq S=N$ and $R_i = \emptyset$ as $\beta = 0$, which implies that $S$ is individually rational. Therefore $S$ is the unique maximum stable invitation. 
\end{proof}

While Lemma~\ref{lemma:trivial_poly_gsip} is trivial, it is worth noting that the interests of the agents and organizer align in this case, which enables us to efficiently find a maximum stable invitation.

We summarize our hardness and easiness results for the non-strategic case of \GSIP\ in Table~\ref{tbl:gsip_summary}. P denotes the existence of polynomial time algorithms and NP-C (NP-completeness) denotes the hardness result. The entries in boldface remark the results we presented in this work, while the other entries are directly implied by our results.\footnote{We note that the single entry (P$^*$) can be derived from the result of Darmann et al.~\cite{GASP2012WINE}.}

We also emphasize that our NP-hardness results hold for both finding maximum stable invitations and finding maximum individually rational invitations (which may not be envy-free). Finding the maximum envy-free invitation is trivial because the full invitation is always envy-free. 

\begin{table}[!ht]
	\centering
\begin{tabular}{|l|*{6}{c|}}\hline
\multirow{2}*{} & \multicolumn{3}{c|}{Simple Preferences} & \multicolumn{3}{c|}{Arbitrary Preferences} \\ \cline{2-7}
&\makebox[3em]{$\beta=0$}&\makebox[3em]{$\beta=1$}&\makebox[3em]{$\beta\geq 2$} 
&\makebox[3em]{$\beta=0$}&\makebox[3em]{$\beta=1$}&\makebox[3em]{$\beta\geq 2$} \\ \hline
$\alpha = 0$ & \textbf{P} & P & \textbf{NP-C} & \textbf{P$^*$} & \textbf{P} & NP-C\\ \hline
$\alpha = 1$ & \textbf{P} & \textbf{NP-C} & NP-C & \textbf{P} & NP-C & NP-C\\ \hline
$\alpha \geq 2$ & \textbf{P} & NP-C & NP-C & \textbf{NP-C} & NP-C & NP-C\\ \hline
\end{tabular}
\caption{Computational complexity of finding maximum stable invitations and finding maximum individually rational invitations given an instance of \GSIP. }
\label{tbl:gsip_summary}
\end{table}


\subsection{The Strategic Case} \label{sec:gsip_strategic}
Recall that our impossibility result for \ASIP\ in the strategic case immediately implies the same impossibility result for \GSIP\ (with arbitrary preferences) because \ASIP\ is a special case of \GSIP\ by restricting the acceptance set and rejection sets to empty sets. 
However we obtain an impossibility result for another special case of \GSIP\ in which we assume that all agents have simple preferences, as stated in Lemma~\ref{lemma:gsip_impossible}. This strengthens the hardness results in the strategic case of \GSIP. Although we do not formally define a mechanism and strategy-proofness in the context of \GSIP, the reader should assume definitions analogous to the one provided in Section~\ref{sec:def_notation_asip_mechanism}.

\begin{Lemma} \label{lemma:gsip_impossible}
There is no strategy-proof mechanism that finds a stable invitation when given an instance of \GSIP, even if we assume that all agents have simple preferences.
\end{Lemma}
\begin{proof} Consider a $(0,1)$-instance of \GSIP\ with two agents $N = \{a_1, a_2\}$.
We assume that these agents have simple preferences and their rejection sets are given by:
\begin{table*}[!h]
	\centering
\begin{tabular}{cc}
	$R_1 = \{a_2\}$ ,  & $R_2 = \{a_1\}$.
\end{tabular}
\end{table*}

Given $(R_1, R_2)$, the only two stable invitations are $S_1 = \{a_1\}$ and $S_2 = \{a_2\}$. Suppose that a mechanism chooses $S_2$ given $(R_1, R_2)$. Then $a_1$ has an incentive to misreport as if her rejection set were empty (i.e.\ $\hat{R}_1 = \emptyset$).

Given $(\hat{R}_1, R_2)$, the only stable invitation is now $S_1$ (and $S_2$ is no longer envy-free with respect to $a_1$). Since $S_1 \succ_1 0 \sim_1 S_2$, agent $a_1$ has an incentive to misreport her rejection set.
Similarly, if the mechanism were to choose $S_1$ given $(R_1, R_2)$, then $a_2$ would have an incentive to misreport with $\hat{R}_2 = \emptyset$. Therefore there is no strategy-proof mechanism that is capable of finding a stable invitation. 
\end{proof}

In summary we have shown that the General Stable Invitation Problem (\GSIP) is both computationally and game-theoretically hard to solve, even if we assume simple preferences of agents (which is a sub-class of \GSIP). It is worth noting that the only positive result we obtained for the strategic case is when all agents have anonymous \INC-preferences (i.e.\ when given an \INC-instance of \ASIP). Yet we believe that there are many other interesting ways to tackle the problem despite our negative result, as we discuss in the next section.

\section{Contributions and Future Work} \label{sec:Discussion}

The main contribution of this work is a thorough analysis of the Stable Invitation Problem from both the computational complexity perspective and game-theoretic perspective. 
We first defined the Anonymous Stable Invitation Problem (\ASIP) in which agents have anonymous preferences over sizes of invitations. 
We showed that the problem of finding a maximum stable invitation is solvable in polynomial time when agents are truthful, and showed that it is in general impossible to design a strategy-proof mechanism when agents are strategic. Yet we also discovered that when the interests of the agents and organizer align, the strategic case of \ASIP\ is solvable, by presenting a strategy-proof mechanism. 

We then relaxed the assumption of anonymous preferences and introduced acceptance sets and rejections sets by which agents can specify their constraints. We formally defined the General Stable Invitation Problem (\GSIP), and showed that in general it is NP-hard to find a maximum stable invitation even if we assume truthful agents with simple preferences.  We also showed that the problem is still NP-hard when we only seek individually rational invitations (that may not be envy-free). Yet we presented efficient algorithms that find a maximum stable invitation for some special instances of \GSIP\ with arbitrary preferences. In the strategic case, an impossibility result for \ASIP\ immediately implies the same negative result for \GSIP.  However, we also obtained an impossibility result for \GSIP\ with simple preferences in which agents only care about the identities of attendees but not the size of an invitation.

While this work answers many interesting questions with flavors of complexity and game-theory, it also leaves many interesting open problems. In the Anonymous Stable Invitation Problem, we argued that our easiness results can be extended to the case with multiple alternatives for the time of the event. The underlying assumption we made is that the organizer needs to pick a specific alternative for the time. However one can think of a different setting where the same event can be held multiple times (for example, multiple receptions during a long conference), and our results are not immediately applicable to this setting. This generalization leads to stable matching problems, which is an interesting direction for future work.  In addition, our impossibility result is obtained partly due to the conflict between strategy-proofness and capability of finding a solution; one can instead seek a strategy-proof mechanism that is guaranteed to find an approximate solution, which is another interesting direction.

In the General Stable Invitation Problem, we showed that the decision problem is computationally hard to solve, but there has been much research devoted to developing approximation algorithms. Approximation algorithms are interesting of their own, but we are more interested in designing efficient mechanisms that are strategy-proof and are guaranteed to find an approximate solution (with respect to the maximum stable solution). An interesting question is to ask whether it is possible to design such mechanisms, and how good the approximation factors can be.
Another approach is to relax the stability requirements and let both individual rationality and envy-freeness can be violated to a small extent, while we seek a maximum solution. This leads to an optimization problem with soft-constraints (instead of hard-constraints), and such problems are often solvable in polynomial time. We can further investigate how one can design strategy-proof mechanisms in such settings, and this is yet another direction for future work. 
Finally we can generalize the Stable Invitation Problem even further to the weighted version, in which the organizer has a weight vector over agents such that the organizer tries to maximize the sum of the weights of attendees instead of the number of attendees.

\subsection*{Acknowledgments}
This work was funded in part by the National Science Foundation (under grants IIS-0963478 and IIS-1347214), the U.S. Army (under grant W911NF1010250), and the Kwanjeong Educational Foundation.



%
\begin{contact}
Hooyeon Lee\\
Computer Science Department\\
Stanford University\\
Stanford, USA\\
\email{haden.lee@stanford.edu}
\end{contact}

\begin{contact}
Yoav Shoham\\
Computer Science Department\\
Stanford University\\
Stanford, USA\\
\email{shoham@cs.stanford.edu}
\end{contact}




\begin{thebibliography}{1}

\bibitem{GASP2012WINE}
Darmann, Andreas and Elkind, Edith and Kurz, Sascha and Lang, J{\'e}r{\^o}me and Schauer, Joachim and Woeginger, Gerhard.
\newblock Group Activity Selection Problem. 
\newblock {\em Internet and Network Economics (WINE-12)}, pages 156-169, Springer Berlin Heidelberg, 2012.

\bibitem{von1947theory}
Von Neumann, John and Morgenstern, Oskar.
\newblock The Theory of Games and Economic Behavior.
\newblock {\em Princeton University Press}, 1947.

\bibitem{Garey_Max_Is_Cubic}
Garey, Michael R., David S. Johnson, and Larry Stockmeyer.
\newblock Some simplified NP-complete graph problems.
\newblock {\em Theoretical Computer Science}, 1(3):237-267, 1976.


\bibitem{sen1998formal}
Sen, S. and Durfee, E.H.
\newblock A formal study of distributed meeting scheduling.
\newblock {\em Group Decision and Negotiation}, 7(3):265-289, 1998.

\bibitem{ephrati1994non}
Ephrati, E. and Zlotkin, G. and Rosenschein, J.S.
\newblock A non-manipulable meeting scheduling system
\newblock {\em Proceedings of the 13th international workshop on distributed artificial intelligence}, pages 105-125, 1994.

\bibitem{lovasz_local_lemma}
Shearer, James B.
\newblock On a problem of Spencer
\newblock {\em Combinatorica}, 5(3):241-245, 1985. 




\bibitem{BogomolnaiaJackson}
Bogomolnaia, Anna, and Matthew O. Jackson.
\newblock The stability of hedonic coalition structures.
\newblock {\em Games and Economic Behavior}, 38(2):201-230, 2002. 

\bibitem{DrezeGreenberg}
Dreze, Jacques H., and Joseph Greenberg
\newblock Hedonic coalitions: Optimality and stability.
\newblock {\em Econometrica: Journal of the Econometric Society}, pages 987-1003, 1980. 

\bibitem{AzizBrandl}
Aziz, Haris, and Florian Brandl
\newblock Existence of stability in hedonic coalition formation games
\newblock {\em Proceedings of the 11th International Conference on Autonomous Agents and Multiagent Systems}, Vol 2. page 763-770, 2012. 




\bibitem{GroupPurchasing}
Lu, Tyler, and Boutilier, Craig. 
\newblock Matching Models for Preference-sensitive Group Purchasing
\newblock {\em Proceedings of the 13th ACM Conference on Electronic Commerce}, pages 723-740, 2012.

\bibitem{AlgorithmDesign}
Kleinberg, Jon, and Tardos, {\`E}va.
\newblock Algorithm design.
\newblock {\em Pearson Education India}, 2006.

\bibitem{NP_hedonic}
Ballester, Coralio.
\newblock NP-completeness in hedonic games
\newblock {\em Games and Economic Behavior}, 49(1):1-30, 2004.

\bibitem{EdmondsKarp}
Edmonds, Jack, and Richard M. Karp.
\newblock Theoretical improvements in algorithmic efficiency for network flow problems.
\newblock {\em Journal of the ACM (JACM)}, 19(2):248-264, 1972.

\end{thebibliography}
\end{document}